\documentclass[11pt]{article}

\usepackage{amssymb,amsmath,amsfonts,latexsym,bbm,ae,aecompl}
\usepackage{color,graphicx,enumitem,subfigure}

\newcommand{\qed}{\hfill $\square$ \smallbreak}
\newenvironment{proof}{\noindent\textbf{Proof:}}{\qed}

\newcommand{\FF}{\vspace*{\medskipamount}}
\newcommand{\FFF}{\vspace*{\bigskipamount}}

\newcommand{\BBB}{\vspace*{-\bigskipamount}}

\newcommand{\mN}{\mathbb{N}}
\newcommand{\Paragraph}[1]{\BBB\paragraph{#1}}

\newlength{\pagewidth}
\newlength{\captionwidth}
\newtheorem{theorem}{Theorem}
\newtheorem{lemma}{Lemma}
\newtheorem{definition}{Definition}

\newtheorem{proposition}{Proposition}
\newtheorem{corollary}{Corollary}

\begin{document}

\baselineskip           	3ex
\parskip                	1ex

\title{Stability of Adversarial Routing with Feedback\,\footnotemark[1]\FFF\FFF}

\author{	Bogdan S. Chlebus\,\footnotemark[2] 
		\and
		Vicent Cholvi\,\footnotemark[3]
		\and
		Dariusz R. Kowalski\,\footnotemark[4]\FFF\FFF }
	
\footnotetext[1]{The results of this paper were announced in a preliminary form 
				in~\cite{ChlebusCK-NETYS13} and in its final form in~\cite{ChlebusCK15}.}
		
\footnotetext[2]{
		Department of Computer Science and Engineering,
                	University of Colorado Denver,
		Denver, Colorado, USA. 
		Supported by the NSF Grant 1016847.}

\footnotetext[3]{
		Department of Computer Science, 
		Universitat Jaume I, 
		Castell\'on, Spain. 
		Supported by the MEC Grant TIN2011-28347-C01-02.}
			
\footnotetext[4]{
		Department of Computer Science,
               	University of Liverpool,
               	Liverpool, UK. 
		Supported by the EPSRC grant EP/G023018/1.}

\date{}

\maketitle

\FF


\begin{abstract}
We consider the impact of scheduling disciplines on performance of routing in the framework of adversarial queuing.
We propose an adversarial model which reflects stalling of packets due to transient failures and explicitly incorporates feedback produced by a network when packets are stalled.
This adversarial model provides a methodology to study stability of routing protocols when flow-control and congestion-control mechanisms affect the volume of traffic.
We show that any scheduling policy that is universally stable, in the regular model of routing that additionally allows packets to have two priorities, remains stable in the proposed adversarial model.

\vfill

\end{abstract}

\vfill

\thispagestyle{empty}
\setcounter{page}{0}
\newpage

\section{Introduction}

We consider routing in communication networks when transient transmission failures and congestion-control mechanisms affect the number of packets handled by nodes. The overall goal is to have fluid traffic  and avoid a congestive collapse, as means to optimize the use of a network.

The following are general comments relevant to the methodological approach considered in this work.
First, routing protocols  operate in environments involving flow control and congestion control.
In the real-world implementations of routing, packets are dropped when the time assigned for processing has been surpassed, which has a stabilizing effect on performance.
The algorithms we consider do not drop packets before they are delivered to their destinations.
This requires assuming potentially unbounded private memory at nodes to accommodate any number of packets in transit.
Second,  stochastic approaches to model performance usually rely on strong assumptions about how traffic is generated, which may be considered too constrained.
In spite of that, stochastic approaches  encounter technical challenges to assess factors underlying a network's performance~\cite{Srikant-book2004}. 
In this paper, we consider worst-case bounds on the performance of routing algorithms in adversarial frameworks of packet generation.
This allows one to obtain insights into traffic efficiency phenomena while avoiding making ad-hoc stochastic assumptions and abstracting from low level mechanisms implemented on the network and transport layers.

Adversarial queuing was proposed as a methodology to analyze worst-case bounds on routing performance in a framework of traffic environment determined by parameters like injection rates and burstiness~\cite{AndrewsAFLLK01,BorodinKRSW01}. 
The basic aspect of a satisfactory behavior of a system in adversarial queuing is stability, which is defined to mean that  the number of packets handled simultaneously is bounded at all times.

Our goal is to extend the model of adversarial queuing by incorporating features representing congestion control.
This includes the feedback provided by the network to the nodes to notify them that some packets have been stalled.
Delayed feedback is more realistic than assuming that failures are known in advance when an execution of a routing protocol starts. 

The routing protocols that we consider do not drop packets intentionally, but still some packets may be delayed due to malfunctioning of a communication infrastructure.
This includes transient wire-link failures and interferences on wireless links, to the effect that packets may  fail to be successfully transmitted between nodes. 

Another related situation occurs when scheduled packets are delayed due to nodes being switched off for energy savings. 
The recently adopted IEEE 802.3az Energy Efficient Ethernet (EEE) standard, as described in~\cite{IEEE-EEE}, is expected to be conducive to energy savings in local area networks by implementing mechanisms to have nodes temporarily unavailable to cooperate in routing~\cite{ChristensenRN10}.


\Paragraph{Our contributions.}

Now we summarize the contributions of this paper.
We propose an adversarial model to study routing in faulty systems, which incorporates feedback about packets that are stalled on their itineraries. 
The model is an extension of leaky-bucket regulations to model injecting packets into a system subject to constraints on injection rate and burstiness.
In the proposed adversarial model, the adversary learns of transient failures on links that caused packet delay after some time rather than instantaneously, which represents interaction  of routing protocols with congestion-control mechanisms.
Knowledge of the adversary is interpreted as an additional constraint on the power to inject packets.
Protocols are considered stable when they are universally stable and unstable otherwise.
We show that any scheduling policy that is unstable under the regular adversarial queueing model remains such in the proposed adversarial model.
Next, we demonstrate that any scheduling policy universally stable in the $2$-priority model remains stable in the adversarial model that we propose. 
This extends our understanding of universal stability in adversarial queuing with prioritized packets,  as obtained by \`{A}lvarez et al.~\cite{AlvarezBDSF05}.
Finally, we show a possible extension of the proposed adversarial model to capture permanent link faults.


\Paragraph{The related work.}

The adversarial methodology to study store-and-forward routing in wired networks was proposed by Andrews et al.~\cite{AndrewsAFLLK01} and Borodin et al.~\cite{BorodinKRSW01}.  
Adversarial communication in wireless networks was considered by Andrews and Zhang~\cite{AndrewsZ07}.
Stability of broadcast protocols in adversarial multiple-access channels was studied by Anantharamu et al.~\cite{AnantharamuC15,AnantharamuCKR-JCSS19}, Bender et al.~\cite{BenderFHKL05}, and Chlebus et al.~\cite{ChlebusKR-DC09,ChlebusKR-TALG12}.

Adversarial models capturing failures have been proposed in the literature in various network settings. 
Borodin et al.~\cite{BorodinOR04} considered slowdowns associated with links.
The papers~\cite{BlesaCFLMSST09,Cholvi08,KoukopoulosMS-TCS07,KoukopoulosMS-JPDC07} considered dynamic changes in the link capacities, with the intention to interpret such transient decreasing of capacity of a link as a transient failure of the link.

\'{A}lvarez et al.~\cite{AlvarezBDSF05} proposed a model that allows transient disruptions of the connectivity of the system. 
That model was extended by \'{A}lvarez et al.~\cite{Alvarez-TCS11} to incorporate node failures. 
The models mentioned above assume that the adversary can make a link fail at any round.
The papers~\cite{AlvarezBDSF05,Alvarez-TCS11,LimJA12} assume that, at each round, the adversary knows when links fail. 
It is then natural for adversaries to be equipped with the power to adjust the injection of packets to such  events, possibly even before link failures occur.
We propose an approach in which the constraints on the adversary, in terms of the injection rate and burstiness, are modified after some time delay triggered by malfunctioning of links.

The adversarial approach has been applied to modeling malfunctioning of wireless networks, including single-hop multi-channels.
Bhandari and Vaidya~\cite{BhandariV07,BhandariV10} considered broadcast protocols in multi-hop wireless networks with nodes prone to failures.
Gilbert et al.~\cite{GilbertGKN-INFOCOM09} considered a multi-channel where the adversary controls how information flows on subsets of channels.
Meier~et al.~\cite{MeierPSW09} considered adversarial multi-channel single-hop networks when  some $t$ channels out of $m$ could be disrupted in a round, with $m$ known and $t$ not known. 

Adversarial queuing was applied when studying interference and jamming in wireless networks, including single-hop multi-channels and multiple access channels.
Lim et al.~\cite{LimJA12} proposed an adversarial model to capture interferences among the links in wireless networks. 
In this case, at each round the adversary assigns specific edge rate vectors that are assumed to keep the network stable.
These vectors can be interpreted as reflecting the degree to which edges fail by not providing their full capacity.
Anantharamu et al.~\cite{AnantharamuCKR-JCSS19} considered multiple access channels with adversarial jamming, when the attached stations perceive jamming as colliding attempts by different stations to access the channel.
Awerbuch~et al.~\cite{AwerbuchRS-PODC08} studied saturation throughput of randomized protocols  in adversarial multiple access channels subject to jamming.
Gilbert~et al.~\cite{GilbertGN09} studied single-hop multi-channel networks with communication subject to adversarial jamming.

For a general discussion of topics related to the mechanisms of flow and congestion control, see~\cite{MamatasHT07,Srikant-book2004}.

\section{The adversarial model with feedback}

\label{sec:model}

We propose an adversarial approach to study stability of routing which captures packet delays due to failures of network elements.
This methodology is an extension of the regular leaky-bucket adversarial model determined by injection rate and burstiness.
The new component is the feedback from the network after an unsuccessful transmission. This feedback restricts the adversary's capability to inject packets.

We will model networks as directed graphs $G = (V,E)$, where the vertices in~$V$ represent the nodes of the network and the edges in~$E$ are the links connecting nodes. 
The orientation of an edge represents the direction in which the link can transmit data.
The networks we consider are \emph{synchronous}, in that an execution of a communication protocol is partitioned into rounds.

Each packet is injected by the adversary into some node and assigned a path through the network to traverse. 
Such paths cannot contain the same link more than once. 
If more than one packet wishes to cross an edge~$e$  in a round,  then a routing protocol chooses one of these packets to send across~$e$, while the remaining packets are kept in a queue at the tail of the edge~$e$. 
When a packet reaches the destination node then it is absorbed, which means that it disappears.

A packet travels through the network with additional information associated with it, like the destination address or the round when it was injected into the network.
A packet encapsulated with this information makes an atomic unit of data to be transmitted through links, which we call simply a \emph{message}.
Messages and rounds are scaled to each other, in that it takes one round to transmit a message through a link.

In this work, we consider routing environments in which a packet scheduled to be transmitted over a link in a round may fail to traverse the link, this possibly happening for a consecutive number of rounds at a time. 
When we use the terms  ``faulty round'' and ``faulty link,'' these refer to situations where failures in messages traversing links occur.
We assume that messages are never lost in transmissions, in that they are successfully handed over from a node to the next neighbor on the traversed path, until the packets reach their  destination.

\subsection{A leaky-bucket regulation}

We define the adversarial model by how traffic is regulated.
We use a traffic descriptor using the notion of a leaky-token bucket, as proposed by Turner~\cite{Turner02}.
Traffic demand is determined by packets injected into the network, each packet being assigned a path to traverse.
The notion of packets becoming stalled in their journeys is to represent a general malfunctioning of the system that results in a packet getting delayed, when an attempted transmission on a link fails, regardless of what is the reason of failure.
A packet is \emph{stalled} in round $t$ if it is attempted to be transmitted but the transmission fails.

An adversary is defined by three parameters: \emph{injection rate~$r$}, such that $0<r \le 1$, \emph{burstiness~$b$}, which is a positive integer, and \emph{feedback delay~$\delta$}, which is also a positive integer.
These three parameters together determine the \emph{adversarial type~$(r,b,\delta)$}.

We will consider two kinds of virtual objects called tokens and antitokens.
A \emph{token} is in a bucket and represents the ability to inject a packet.
An \emph{antitoken} represents a stalled packet, and so the need to  decrease the traffic by one packet to avoid congestion.

A single \emph{bucket} is a variable~$K$ storing a number; $K$ is initialized to~$0$.
When $K\ge 0$, then~$\lfloor K \rfloor$ is interpreted as the number of tokens in the bucket.
The bucket's capacity is~$b$.
In general, $K$ may assume negative values; in such a case the bucket does not contain any tokens.
The operations performed on $K$ are as follows.
\begin{enumerate}
\item[1)] 
In each round, first $r$ is added to the bucket~$K$ by the assignment $K\leftarrow K+r$.
If at this point $K>b$ then $K$ is immediately modified by the assignment $K\leftarrow b$, which is interpreted as the bucket's overflow.

\item[2)] 
The adversary injects some $i$ packets into the system and simultaneously removes $i$ tokens from the bucket $K$ by performing $K\leftarrow K-i$.
For this to be possible to perform, the inequalities $0<i\le K$ need to hold.

\item[3)] 
The adversary controls when packets are stalled. Each round a packet is stalled, an antitoken $g$ is created, carrying a value. 
The value of a newly created antitoken is initialized to~$\delta$.
The value of an antitoken gets decremented by~$1$ in each round.
An antitoken disappears in two possible ways, as decided by the adversary.
One results in removing the antitoken and simultaneously modifying $K\leftarrow K-r$ while the token's value is still positive.
Another is when the value becomes~$0$, then the antitoken disappears, and simultaneously $K\leftarrow K-r$; this represents the maximally delayed feedback. 
\end{enumerate}

Intuitively, when the adversary decides to annihilate a token, then this represents a moment when the adversary obtains feedback from the network of a stalled packet.
One token and one antitoken disappear simultaneously, as by annihilation resulting from their getting mixed together, which explains the terminology.

This completes the specification how a single bucket operates, when it is considered in isolation independently from other buckets.
The complete picture is such that we associate a bucket~$K_e$ with each directed edge~$e$ of the network.

The operations on these buckets are coordinated as follows:
\begin{enumerate}
\item[1)] 
Every bucket gets incremented by $r$ in each round, subject to possible overflow which makes a bucket store precisely $b$ tokens.

\item[2)] 
A packet has a path assigned to traverse; when a packet gets stalled then an antitoken is created for each bucket~$K_e$ associated with an edge~$e$ that the packet is still to traverse; all these antitokens are said to be \emph{related}.

\item[3)] 
When the adversary injects a packet to traverse a path, then it removes a token from each bucket~$K_e$ associated with any edge~$e$ of the path.
For this to be possible to be performed, each bucket on the path needs to include at least one token.

\item[4)] 
When the adversary destroys an antitoken~$g$ created by a stalled packet~$p$ on some link, then such a destruction, and the matching operation $K\leftarrow K-r$, is performed on each related token $g'$ created on a link which $p$ was still to traverse when $g$ was created, and the bucket associated with a link that $p$ was still to traverse when $g$ was created.
\end{enumerate}
This completes the specification of how the adversary can inject packets into the network.

\subsection{Comparison with the regular adversary}

The regular leaky-bucket adversary is defined by two parameters: \emph{injection rate~$r$}, such that $0<r \le 1$, and a positive integer \emph{burstiness~$b$}.
These two parameters together determine the \emph{adversarial type~$(r,b)$}.


\begin{lemma}
\label{lem:comparing-adversarial-models}

A delayed-feedback leaky-bucket adversary of type $(r,b,\delta)$ is at least as powerful as the regular leaky-bucket adversary of the type $(r,b)$.
\end{lemma}

\begin{proof}
We compare the two adversarial models as regulated by a leaky bucket of tokens.
When the  delayed-feedback adversary of the type $(r,b,\delta)$ does not induce any stalling among the packets, then the regulatory properties of a bucket of tokens determine the regular adversary of the type $(r,b)$.
\end{proof}


\begin{theorem}
\label{thm:from-regular-to-delayed}

If a scheduling policy is unstable in a network $G$ under a scheduling policy $\cal S$ against the regular adversary of type~$(r,b)$ then this same scheduling policy~$\cal S$ is unstable in the network $G$ against a delayed-feedback adversary of the type~$(r,b,\delta)$, for any positive integer~$\delta$.
\end{theorem}

\begin{proof}
Consider an unstable execution of routing against the adversary of the type $(r,b)$ when the scheduling policy~$\cal S$ is applied.
A similar unstable execution can be produced by the adversary of the type $(r,b,\delta)$, by Lemma~\ref{lem:comparing-adversarial-models}.
\end{proof}

We consider the following specific scheduling policies: First-In-First-Out (FIFO),  Nearest-To-Go (NTG), Farthest-From-Source (FFS), and Slowest-Previous-Link with ties broken using the Nearest-From-Source policy (SPL-NFS).


\begin{corollary}
\label{corollary:specific-unstable-policies}

Each of the scheduling policies FIFO, NTG, FFS and SPL-NFS is unstable in some network against a delayed-feedback adversary with injection rate less than~$1$.
\end{corollary}

\begin{proof}
The argument is based on Theorem~\ref{thm:from-regular-to-delayed}.
We rely on the respective instability results obtained for the regular adversary.
The instability of the scheduling policies FIFO, NTG, and FFS follows from the  instabilities obtained by Andrews et al.~\cite{AndrewsAFLLK01}, and the instability of SPL-NFS follows from the related result given by Blesa et al.~\cite{BlesaCFLMSST09}.
\end{proof}

\section{Properties of the adversarial model with feedback}
\label{sec:implementation}

In this section, we investigate properties of the adversarial model with feedback presented in the previous section. 
A key point of the model is the concept of an antitoken, which represents a stalled packet and thus the need to decrease the bound on the future traffic by one packet to avoid congestion. 
We re-define the adversary by expressing its power in a more analytical way to facilitate the future technical analysis.
More precisely, we describe our adversary in terms of an admissibility condition, which involves a delay function accounting for the impact of each antitoken's annihilation on decreasing the amount of traffic that could be injected.

\subsection{Reformulation of the adversarial model}

In each round, the adversary may inject packets into some of the nodes in the network. 
In order for stability to be achievable in principle, the adversary needs to be somehow restricted.
Such restrictions imposed on the regular leaky-bucket adversary are represented by its adversarial type $(b,r)$, where $b \geq 1$ is a natural number and $r$ satisfies $0 \leq r < 1$. 
For each link, the injection rate~$r$  models an upper bound on the frequency with which packets that eventually need to traverse the link can be injected into the network. 
The burstiness~$b$ represents the maximum number of packets, among those that need to traverse the same link, that the adversary can inject into the network in one round.
The precise interpretation of such a type $(b,r)$ is that in any time interval $\tau$ of length $|\tau|$ the adversary may inject at most $r |\tau|+b$ packets that need to traverse the same edge.
An adversary is free to choose both the source and the destination node for any injected packet. 
The adversary also determines the individual path from the source to the destination that any specific packet needs to traverse.


\begin{figure}[t]
\begin{center}
\includegraphics[width=0.6\textwidth, angle=270]{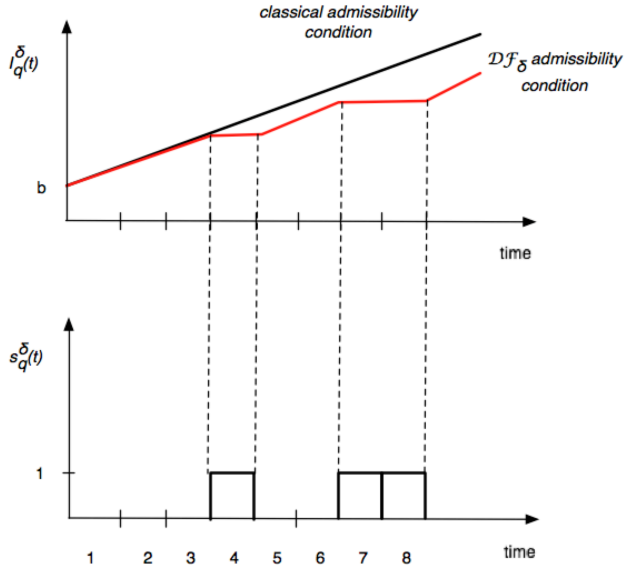}
\FFF\FFF

\parbox{\captionwidth}{
\caption{\label{fig_DF_example}
An illustration for the admissibility condition~\eqref{eqn:admissible}.
The number~$s_q^{\delta}(t)$ represents the rounds when the adversary reacts to stalled packets in $q$. 
The injected data~$I_q^{\delta}(t)$ represents the admissible amount of data that the adversary can inject due to the extra rounds incurred by stalling.}}
\end{center}
\end{figure}

An extension to the adversarial model with delays, which we refer to as $\cal DF_{\delta}$ (\emph{feedback delayed by up to~$\delta$}),  is as follows:

Let $I^{\delta}_{q}(t)$ represent the total number of packets which the adversary injects at round~$t$ that have queue $q$ on their path. 
We say that the packet injections are {\em admissible for rate~$r$ and burstiness~$b$} if the following \emph{admissibility condition} holds for all $q$:
\begin{equation}
\label{eqn:admissible}
\sum_{t \in T} I^{\delta}_{q}(t) \leq r \sum_{t \in T} (1 - s^{\delta}_{q}(t)) + b,
\end{equation}
where $T$ represents a contiguous time interval and $s_q^{\delta}$ is a delay function. A formal specification of $s_q^{\delta}$ is provided in Figure~\ref{figure:specification-of-reactive-function} in Section~\ref{subsec:functions}. 

The function $s_q^{\delta}$ represents the rounds when the adversary becomes constrained by  packets getting stalled in queue~$q$.
We interpret them as the rounds where the adversary ``reacts'' to stalled packets.  
At any round~$t$, the adversary takes into account the values of $s_q^{\delta}(t')$, for $t' \leq t$, based only on the notifications received up to that round, as reflected in Figure~\ref{figure:specification-of-reactive-function}. 
This means that the admissibility condition in Equation~(\ref{eqn:admissible}) refers only to the received notifications and the packets injected into the system. 
Figure~\ref{fig_DF_example} gives an example of the admissibility condition in a queue~$q$. 

Next, we compare the adversarial power in $\cal DF_{\delta}$  with its \emph{matching adversary}  in $(r,b,\delta)$, that is, adversaries with the same respective parameters.


\begin{proposition}

Any adversary in $\cal DF_{\delta}$  is equivalent to its matching adversary in $(r,b,\delta)$ for all $0 \leq r < 1$ and $b \geq 1$.
\end{proposition}

\begin{proof}
What distinguishes the adversary in $\cal DF_{\delta}$ and in $(r,b,\delta)$ is that, on the one hand, while in $\cal DF_{\delta}$ the adversary obtains the feedback (regarding stalled packets) provided by the underlying system, in $(r,b,\delta)$ the adversary has  the power to control any malfunctioning of the network.  A second difference is how the adversary can inject packets into the network. While in $(r,b,\delta)$ that is specified by means of a leaky-bucket mechanism; in  $\cal DF_{\delta}$ that is done by means of an admissibility condition.

Regarding the first difference, we have that these two approaches are equivalent, as the adversarial model is to capture a worst-case behavior of the network.

Regarding the second difference, we have that they are also equivalent. Let us first start by considering the admissibility condition. 
It captures the following intuitions:
\begin{enumerate}
\item
When $s_q^{\delta}(t)=0$: 
This can be interpreted as if there are no annihilations of antitokens at time~$t$ in~$q$. 

In this case, the value of the bucket $K$ is not modified due to some annihilation, as also happens in the admissibility condition. 
Since the regulatory properties of the bucket of tokens are equivalent to the regulatory properties of the admissibility condition~\cite{AndrewsAFLLK01}, we are done.

\item
When $s_q^{\delta}(t)=1$: 
This can be interpreted as if an antitoken has been annihilated at time~$t$ in~$q$. 

In this case, the value of the bucket is reduced by $r$, and since $1-s_q^{\delta}(t) =0$, then the component corresponding to time $t$ in the admissibility condition is equal to $0$.
In other words, there is a  reduction of $r$ in the admissibility condition with respect to the case where $s_q^{\delta}(t)=0$. 
The regulatory properties of the bucket of tokens is equivalent to the regulatory properties of the admissibility condition~\cite{AndrewsAFLLK01}, so we are also done.
\end{enumerate}
The reasoning in the case where we consider the leaky-bucket mechanism is similar: an annihilated antitoken at time $t$ in $q$ can be interpreted as the case where $s_q^{\delta}(t)=1$.
\end{proof}

The definition of stability in $\cal DF_{\delta}$ is similar to the definition stated under other adversarial models.


\begin{definition}
Let $G$ be a network, $\cal P$ a scheduling policy and $\cal A$ an adversary of type $(r,b)$. 
Let $\cal D$ be an execution of protocol $\cal P$ against $\cal A$ in $G$. 
For a positive integer $t$, let $Q_{\cal D}(t)$ be the number of packets that are queued in the system at time~$t$. 
Protocol $\cal P$ is \emph{stable} on $G$ against~$\cal A$ if, in each such execution~$\cal D$, all the numbers $Q_{\cal D}(t)$ are bounded. 
Protocol~$\cal P$ is \emph{universally stable} if  it is stable  against any adversary with injection rate $r < 1$ and in any network.
\end{definition}

\subsection{Delay functions and reactive functions}
\label{subsec:functions}

We say that a queue $q$ is \emph{stalled} in round $t$ if  some packet in the queue is stalled in this round.
Next, we introduce the function~$w_q$ which represents the rounds where delays occur at queue~$q$.


\begin{definition}
\label{def_c}

Consider an execution of a system up to round $t$. 
Given a queue~$q$, we define the function~$w_q(t)$ such that $w_{q}(t) = 1$  if the queue $q$~is stalled at round~$t$ and $w_{q}(t) = 0$ otherwise. 
\end{definition}

The rounds that are added to a stalled packet's itinerary occur as a side effect of the adversary's actions. 
The adversary receives information about the extra rounds of stalled packets  occurring at the different queues. 
We consider the case where the adversary becomes constrained by the stalled packets after some time delay; the parameter $\delta$ is used to bound such a maximum delay.

Next, we introduce the notations $T_q$, $D_q^{\delta}$, and $w_{q}^{\delta}$.
We want $w_{q}^{\delta}$ to model the rounds where the adversary becomes constrained by the queue $q$ getting stalled and it also provides the number of notifications.


\begin{definition}
\label{def:function-D}

We will use the following terminology and notations:
\begin{enumerate}
\item
Let $T_q$ be the set of those rounds~$t$ for which $w_q(t)=1$ holds.

\item
For a given function $w_q$, let the function $D_q^{\delta}: T_q\rightarrow \mN$ be such that $D_q^{\delta}(t) = t'$, for $t \in T_q$ and $t \leq t' \leq t + \delta$, where $t'$ is the time when the feedback about the queue $q$ being stalled at time $t$ arrives.

\item
Let $T_q^{\delta}(t)$ be the subset of $T_q$ such that $t' \in T_q^{\delta}(t)$ if $D_q^{\delta}(t') = t$.

\item
For a given $w_q$ and a given $D_q^{\delta}$, let the \emph{delay function $w_{q}^{\delta}$}  be determined by the equality $ w_{q}^{\delta}(t) = | T_q^{\delta}(t)|$.
\end{enumerate}
\end{definition}

Figure~\ref{fig:functions} provides a graphical representation of the notions introduced in Definition~\ref{def:function-D} in a specific example.


\begin{figure}[t]
\begin{center}
\includegraphics[width=0.6\textwidth]{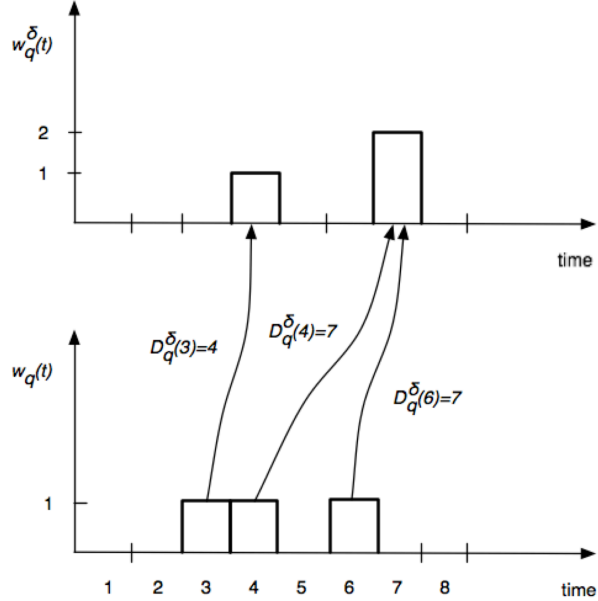}

\parbox{\captionwidth}{
 \caption{ \label{fig:functions}
 An illustration of how the function $D_q^{\delta}$ is determined.}}
 \end{center}
\end{figure}


\Paragraph{Reactive functions.}

When the adversary receives a notification of a queue $q$ getting delayed, then this indicates that some packet in this queue will need an extra round. 
In order to maintain stability, the adversary needs to take such an eventuality  into account. 
A reaction can be interpreted as if temporarily  the injection rate were reduced, which is implemented by the mechanism of tokens and antitokens.
Several notifications could be received at the same time, which needs to be reflected in a cumulative reaction. 
For instance, if at some time $t$ the adversary receives three notifications of stalling for some given link, it will reduce the  long term injection rate of a packet that will cross such a link for three rounds after time~$t$, provided such rounds have not been already reduced because of previous notifications, in which case the next ``available'' rounds will be chosen.


\begin{figure}[t]
\rule{\textwidth}{0.75pt}
\begin{center}
\begin{minipage}{\pagewidth}
\begin{tabbing}
aA\=aA\=aA\=aA\=aaaaaaaaaaaaaA\=aA\= \kill

{\bf initialization:} \\
\\
\> $s_q^{\delta}(t) \leftarrow 0$ for all $t$; \\
 \> $t \leftarrow 1$; \>\>\>\> \emph{\tt \% time in $w_q^{\delta}$} \\
 \>$t' \leftarrow 1$; \>\>\>\> \emph{\tt \% time in $s_q^{\delta}$} \\ 
 \\
 
{\bf repeat forever} \\
\\
\> {\bf if} $w_q^{\delta}(t) \neq 0$ {\bf then} \>\>\>\>  \emph{\tt \% when one or more notifications are} \\
\>\>\>\>\>\> \emph{\tt received at time step $t$} \\ 
\\
\>\> $t' \leftarrow \max(t,t')$; \>\>\> \emph{\tt \% set up the next "available" round} \\
\>\>\>\>\>\> \emph{\tt in $s_q^{\delta}$ after $t$} \\  \\
\> \>  {\bf for} $t_{aux}= t'$ {\bf to} ($t' + w_q^{\delta}(t) -1$) {\bf do} \\ 
\\
\> \> \> \> $s_q^{\delta}(t_{aux}) \leftarrow 1$; \> \emph{\tt \% mark $t_{aux}$ as "reactive"} \\
\\
\>\> $t' \leftarrow t' + w_q^{\delta}(t) -1$; \>\>\> \emph{\tt \% set up the next "available" round in $s_q^{\delta}$} \\
\\
\> $t \leftarrow t + 1$; \>\>\>\>  \emph{\tt \% increment the time step}
\end{tabbing}
\end{minipage}
\FFF

\rule{\textwidth}{0.75pt}
\FF

\parbox{\captionwidth}{
\caption{\label{figure:specification-of-reactive-function}
A specification of the reactive function~$s_q^{\delta}$.
It is determined by a given delayed function $w_q^{\delta}$.
}}
\end{center}
\end{figure}

To formalize this, we specify the reactive function $s_q^{\delta}$ in Figure~\ref{figure:specification-of-reactive-function}.
It is intended to model the rounds when the adversary will react  to a delayed notification of stalling  at a queue~$q$.
This function~$s_q^{\delta}$ is determined by~$w_q^{\delta}$, which gives the rounds when the adversary receives notifications of delays occurring at queue~$q$, as represented by annihilations of antitokens.

\section{Stability  of scheduling policies }
\label{sec:stability}

In this section, we show that some scheduling policies are stable in the adversarial-queuing model~$\cal DF_{\delta}$. 
We will use an auxiliary model, known as the \emph{priority model}, which was introduced in~\cite{AlvarezBDSF05}. 
We refer to the model as a \emph{$c$-priority model} when there are $c$ priorities.
It is obtained by modifying the regular adversarial model~\cite{AndrewsAFLLK01,BorodinKRSW01} so that packets have priorities in the following sense. 
If, at a certain time, more than one packet located at the same queue is ready to be transmitted, then the scheduling policy chooses the packet of the highest priority. 

The following fact provides a relationship between the number of extra rounds due to delays and the number of reactive rounds at a given time interval.


\begin{lemma}
\label{lemma:bound2}
$\sum_{t \in T} w_{q}(t) \leq \sum_{t \in T} s^{\delta}_{q}(t) + \delta$.
\end{lemma}

\begin{proof}
By the definition of function $w^{\delta}_{q}$, for each round $t$ where $w_q(t)=1$, there is a round~$t'$, which is not necessarily different for each~$t$, such that $w^{\delta}_{q}(t')\neq 0$ and $t \leq t' \leq t + \delta$.

From the mechanism used to construct $s^{\delta}_{q}$ and specified in Figure~\ref{figure:specification-of-reactive-function}, the function $s^{\delta}_{q}$ takes the value $1$ as many times as the values taken by the function  $w^{\delta}_{q}$.
It follows that the number of times where $w_{q}$ takes the value $1$ is the same as the number of times where $s^{\delta}_{q}$ takes the value $1$, although not necessarily at the same rounds. 

This means that there exists a bijective increasing function~$rD_q^{\delta}$ such that, for all $t$, when $w_q(t)=1$ then $rD_q^{\delta}(t)=t'$, where $s_q^{\delta}(t')=1$. 

We proceed with considering the following two possible cases.

\noindent
{\sf The case of $t >t'$:}  

This case cannot happen by the construction of $s_q^{\delta}$, because $t_{aux}$ in Figure~\ref{figure:specification-of-reactive-function}  is always at least~$t$.

\noindent
{\sf The case of $t' > t + \delta$:} 

We prove this case by contradiction. Let $t_1$ be the first round such that $rD_q^{\delta}(t_1)=t_1^*$ and $t_1^* > t_1 + \delta$. 
From the construction of $rD_q^{\delta}$ described above and specified in Figure~\ref{figure:specification-of-reactive-function}, we have that  the equality $s_q^{\delta}(t'')=1$ holds  for all $t'' \in [D_q^{\delta}(t_1),rD_q^{\delta}(t_1)-1]$.
Since~$rD_q^{\delta}$ is a bijective increasing function, there must exist some round $t_m$ such that $t_m< t_1$ and $rD^{\delta}(t_m)=t^*_1 -1$. 
Let $t^*_m = t^*_1 - 1$. 

The following two facts hold.
One is that  $t_m < t_1$ and $t^*_m = t^*_1 -1$, which means $t_m^*  > t_m + \delta$. 
The other is that $rD^{\delta}(t_m)=t^*_m$. 
All this contradicts our assumption that $t_1$ is the first round such that $rD_q^{\delta}(t_1)=t_1^*$ and $t_1^* > t_1 + \delta$.
The fact that the inequality $t \leq t' \leq t + \delta$ holds means that for each $t$ such that $w_{q}(t)=1$, the corresponding image in $s^{\delta}_{q}$ will be for a round $t'$ that is delayed by at most $\delta$ rounds. 
We conclude that the inequality $\sum_{t \in T} w_{q}(t) \leq \sum_{t \in T} s^{\delta}_{q}(t) + \delta$ holds. 
\end{proof}

The following Lemma~\ref{ack:simulation1}  shows that if a given scheduling policy is unstable in~$\cal DF_{\delta}$ then it is also unstable in the $2$-priority model.


\begin{lemma}
\label{ack:simulation1}

If a given scheduling policy is unstable against an adversary with injection rate~$r$ and burstiness~$b$ in~$\cal DF_{\delta}$, then such a scheduling policy is unstable against some adversary of  injection rate $r'$ and burstiness $b'$ in the $2$-priority model, where $0<r'<1$.
\end{lemma}

\begin{proof} 
Let us take an adversary $\cal A$ in $\cal DF_{\delta}$ with parameters $(r,b)$. 
Then, according to the admissibility condition in Equation~\eqref{eqn:admissible} and by Lemma~\ref{lemma:bound2}, the following estimates hold:
\begin{eqnarray*}
\sum_{t \in T} I^{\delta}_{q}(t) & \leq & r \sum_{t \in T} (1 - s^{\delta}_{q}(t)) + b  \\
& \leq & r \sum_{t \in T} 1 - (\sum_{t \in T} w_{q}(t) - \delta)) + b \\
& \leq & r \sum_{t \in T} (1 -  w_{q}(t)) + r \delta + b  
\ .
\end{eqnarray*}
The right-hand side of this bound equals~$r  \sum_{t \in T} (1 - w_{q}(t)) + b'$ for $b' = r  \delta +b$.
We obtain by  algebraic manipulations the following bound:
\begin{eqnarray*}
\sum_{t \in T} I^{\delta}_{q}(t) + \sum_{t \in T}  w_{q}(t) & \leq & r  \sum_{t \in T} (1 - w_{q}(t))  + b' + \sum_{t \in T}  w_{q}(t) \\
& = &
r  \sum_{t \in T} 1 + (1 -r)  \sum_{t \in T}  w_{q}(t) + b'  \\
& \leq &
r \sum_{t \in T} 1 + \big((1 -r) \sum_{t \in T}  \frac{\tau}{\tau +1}\bigr) + b' \\
& = &
\bigl(\frac{r + \tau}{\tau +1}\bigr)  \sum_{t \in T} 1 + b'  \\
& = &
\frac{r + \tau}{\tau +1} \mid T \mid + \; b'  \\
& = &
r'   \mid T \mid + \; b'  
\ .
\end{eqnarray*}
In this derivation, we used $w_{q}(t) \leq \frac{\tau}{\tau +1}$ and $r' = \frac{r + \tau}{\tau +1}$.
Observe that $0 < r' < 1$.

Consider an adversarial pattern for injection rate $r$ and burstiness $b$ that results in unstable execution of the given scheduling policy.
Based on the obtained estimate on 
\[
\sum_{t \in T} I^{\delta}_{q}(t) + \sum_{t \in T}  w_{q}(t)\ , 
\]
we define the corresponding adversarial pattern in the $2$-priority model, as specified in the claim of the lemma. 
The constructed specific adversarial behavior follows the same injection pattern as defined by the adversary ${\cal A}$ in $\cal DF_{\delta}$, with a low priority given to all these packets, and additionally it injects a high priority packet at the starting queue $q$ in each round $t$ such that $w_{q}(t)=1$.

Consider the execution of the original scheduling policy under the defined adversarial pattern in the $2$-priority adversarial model. 
By the inequality 
\[
\sum_{t \in T} I^{\delta}_{q}(t) + \sum_{t \in T}  w_{q}(t)  \leq r'   | T | + \; b'\ ,
\]
which holds in the execution in $\cal DF_{\delta}$, we obtain that  queue-congestion of the injected packets, whether of a high or low priority, is constrained by the injection rate~$r'$, with $r' < 1$, and burstiness~$b_v'$. This is because $\sum_{t \in T} I^{\delta}_{q}(t) $ from the original execution in  $\cal DF_{\delta}$ corresponds to the node congestion of the low-priority packets, and $\sum_{t \in T}  w_{q}(t)$ corresponds to the queue-congestion of the high-priority packets, both in the $2$-priority execution.

It remains to argue that the newly defined execution is also unstable in the  $2$-priority model.
We observe that the following invariant holds:
\begin{quote}
\textsf{Invariant:}

There is at most one high-priority packet at a queue in any round in the $2$-priority execution, and 
the transmissions of low-priority packets are the same in both considered executions.
\end{quote}
This fact follows by induction on the round numbers.
Any such a packet is injected in the beginning of each round when a extra round due to stalling occurs in the  execution in $\cal DF_{\delta}$ through some queue. 
Since the execution in $\cal DF_{\delta}$ results in unbounded queues, the other one also does.
\end{proof} 


\begin{theorem}
\label{ackthm:2-priorities}

Any  scheduling policy that is universally stable in the $2$-priority model is universally stable in $\cal DF_{\delta}$.
\end{theorem}

\begin{proof}
Let us suppose it is otherwise, in order to arrive at a contradiction.
This means that there is a scheduling policy $S$ that is universal in the  $2$-priority model but for any burstiness~$b$ there is some rate $r$ such that the inequalities $0<r<1$ hold and such that $S$ is unstable against the adversary with rate $r$ and burstiness~$b$ in $\cal DF_{\delta}$. 
Let us consider such an unstable execution. 
By Lemma~\ref{ack:simulation1}, there is an unstable execution of the scheduling policy~$S$ in the $2$-priority model, for some injection rate $r'$ such that $0<r'<1$ and for burstiness~$b'$. This contradicts the universal stability of~$S$ in the $2$-priority model. 
\end{proof}

\`{A}lvarez et al.~\cite{AlvarezBDSF05} showed that the scheduling policies Farthest-To-Go  (FTG), Nearest-From-Source  (NFS) and Shortest-In-System  (SIS) are universally stable for the $2$-priority model. 
By this and Theorem~\ref{ackthm:2-priorities}, we obtain the following corollary.


\begin{corollary}
\label{cor:stable_policies}
Scheduling policies FTG, NFS and SIS are all universally stable in the adversarial model $\cal DF_{\delta}$.
\end{corollary}

\section{Permanent Failures}
\label{sec:permanent}

In the previous sections it has been assumed that failures occur only in a transient manner. 
In this section, we extend the delayed feedback adversarial model ($\cal DF_{\delta}$) to  incorporate permanent failures into the model of packet injection. 
We denote such an extended  model as $\cal DF^{\it e}_{\delta}$.


\Paragraph{A scenario of failures.}

In $\cal DF^{\it e}_{\delta}$ we consider the case where, in addition to temporary failures, links may also fail in a permanent fail-stop manner, while nodes are not prone to faults. 
We assume that packets are never lost. 

When a link $q$ permanently fails, the adversary will be aware of that in a finite amount of time $\tau'$. That can be modeled by using a notification of the form $\mathit{fail}_q(t)$ to be sent to the adversary when the link fails, which will receive it in at most $\tau'$ rounds.

We remark that $\cal DF^{\it e}_{\delta}$ allows scenarios where both permanent and temporal link failures occur simultaneously. 
Furthermore, a temporary faulty link could also fail in a permanent manner.
For instance, if the link undergoes a transient fault for $\tau$ consecutive rounds, then the failure could be assumed to be permanent, in the sense that the link stays faulty forever.


\Paragraph{An adversarial model.}

In addition to the behavior of the adversary in $\cal DF_{\delta}$, in $\cal DF^{\it e}_{\delta}$ the adversary also reacts to occurrences of permanent link failures. 
This is specified as follows:

\begin{itemize}
\item
Once the adversary receives the notification of a permanent failure, it stops injecting new packets to traverse the respective faulty link. 
Starting from such a notification, the adversary must choose such routes for  injected packets that do not include any permanently faulty link.

\item
It could happen that packets injected prior to a permanent failure notification haven't crossed the failed link. 
In order to make it possible for these packets to reach their destinations, when a packet reaches a node such that the outgoing link to traverse is faulty, the adversary simply re-routes each such a  packets using a new simple path.
Note that in order to re-route a packet, it is necessary that the new path be able to reach the destination, which could impose some restrictions on the failure pattern.
While the new path will not use the same link more than once, it could use links that have been previously used in the original path. 

Furthermore, when the adversary re-routes a packet, it also \emph{updates} the packet with the new path; namely, the concatenation of its original path up to the failed link, and the new path. 
\end{itemize}

Let us  make two observations at this point:
First, at each queue, the scheduling policy will treat in the same way all packets, regardless of whether they are re-routed packets or not.
Second, since re-routed packets are already inside the system, when the adversary re-routes a packet it does not take into account any  additional admissibility condition, other than choosing a simple path.


\begin{lemma}
\label{lem:bou}

The maximum number of re-routed packets in $\cal DF^{\it e}_{\delta}$ is bounded, provided the used scheduling policy is universally stable in $\cal DF_{\delta}$.
\end{lemma}

\begin{proof}
On the one hand, since the scheduling policy is universally stable in $\cal DF_{\delta}$, when a link fails in a permanent manner after failing in a temporary manner the number of waiting packets (at the queue of that link) that will be re-routed is guaranteed to be bounded.
On the other hand, the adversary will stop injecting packets to pass through such a link after receiving the corresponding notification, which will occur in a bounded number of rounds. Then, the number of new injected packets that will be re-routed because of that failure is also bounded. 
Since the number of links is bounded, the overall number of  packets that will be re-routed is bounded.
\end{proof}

Now, we present our main result regarding permanent link failures in $\cal DF^{\it e}_{\delta}$.


\begin{theorem}
\label{the:per}

Any universally stable scheduling policy in $\cal DF_{\delta}$ remains universally stable in~$\cal DF^{\it e}_{\delta}$.
\end{theorem}

\begin{proof} 
Let us consider an arbitrary network and consider an adversarial system execution in $\cal DF^{\it e}_{\delta}$ using an universally stable scheduling policy in $\cal DF_{\delta}$. Let us also consider that, in addition to the original notifications of delays, when a re-routed packet is transmitted at some time $t$ at given queue $q$, a notification $w_q(t)$ is sent to the adversary.
This means  we consider the transmission of re-routed packets as failures.

Suppose the adversary does nothing upon the reception of the notifications for re-routed packets. 
This means that, at each round when the adversary receives a notification, the admissibility condition will allow the adversary to inject one more packet compared with the execution when the adversary reacts to these notifications.

This execution will be equivalent to an execution in $\cal DF_{\delta}$ where the adversary has a larger burst. 
Since by Lemma~\ref{lem:bou} the number of these notifications is bounded, then the increase in that burst is also bounded. Therefore, since this execution is stable in $\cal DF_{\delta}$, it will also be stable in $\cal DF^{\it e}_{\delta}$. 
\end{proof}


\Paragraph{On the use of re-routing for temporary failures.}

Since the use of re-routing has been proposed for permanent failures, a natural question that arises is whether or not such a technique could be used for temporary failures. 

However, the use of re-routing with temporary faulty links presents some problems that may lead to instability. Indeed, re-routed packets are not subject to any admissibility condition regarding the links of the new paths. Therefore, if links disappear and appear over time, that could provoke an accumulation of re-routed packets at some queues, whose occupancy could grow unboundedly.


\begin{figure}[t]
\begin{center}
\includegraphics[scale=0.4]{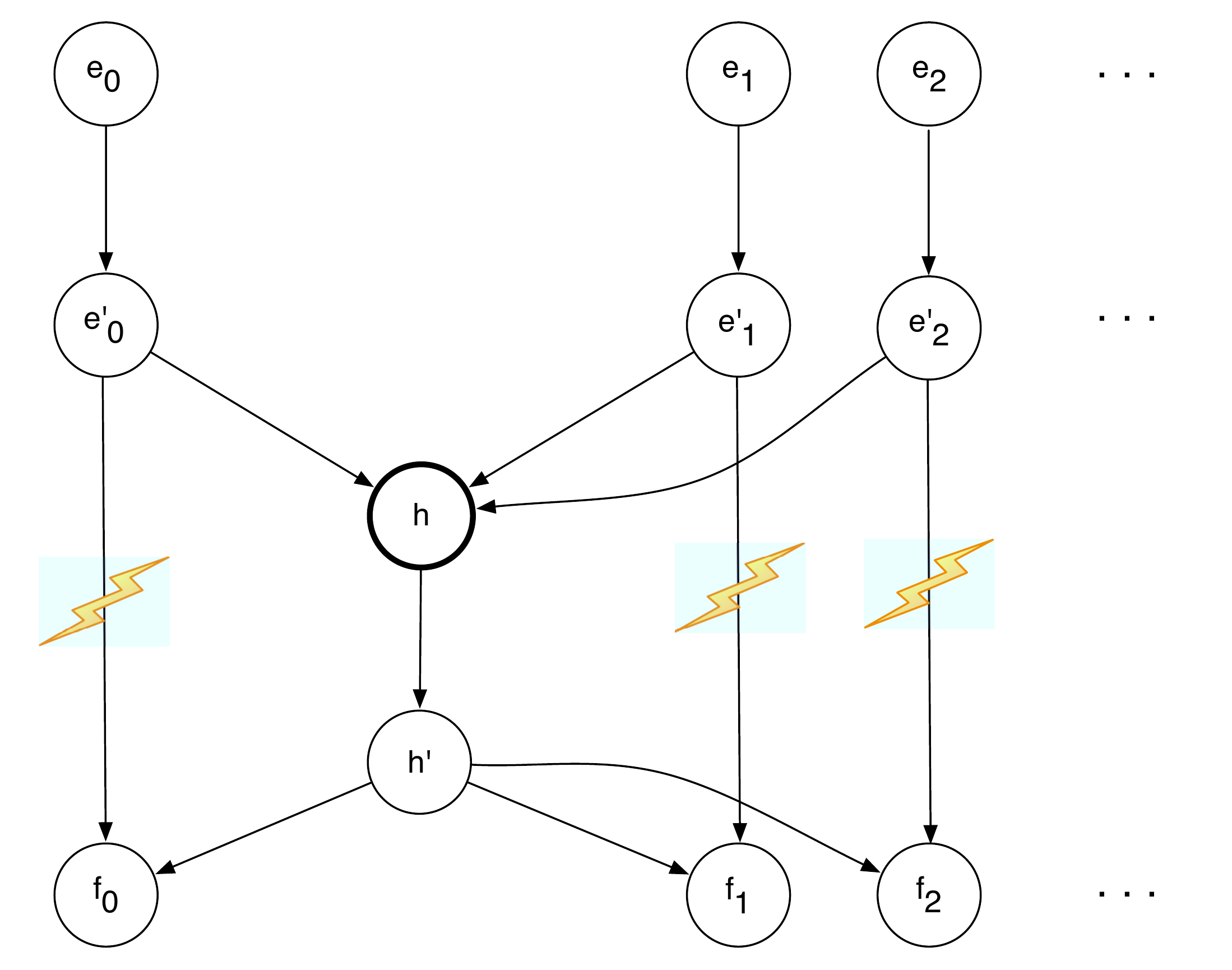}
\FFF\FFF

\parbox{\captionwidth}{\caption{\label{figure-no_rerouting}
The network topology used  in the proof of Lemma~\ref{lem:tem_rou} to show that the use of re-routing with temporary faulty links may lead to instability.}}
\end{center}
\end{figure}


\begin{lemma}
\label{lem:tem_rou}
If temporary faulty links are treated as permanent, any scheduling policy is unstable in $\cal DF^{\it e}_{\delta}$.
\end{lemma}

\begin{proof}
Let us consider the network topology in Figure~\ref{figure-no_rerouting} and assume that links fail in a temporary manner but are treated as permanent so this results in rerouting.
Assume moreover that bursts of $10$ packets are injected at each node $e_i$ with destination $f_i$ and passing through $e'_i$. Assume also that immediately after injecting these packets all links $e'_i \rightarrow f_i$ fail for $10$ rounds, and packets are re-routed through link $h \rightarrow h'$. 
After all packets have been re-routed, faulty links recover and the process is repeated forever. 
It is easy to see that the buffer occupancy of link $h \rightarrow h'$ will grow unbounded.
\end{proof}


\Paragraph{Recovering from permanent failures.}

A consequence of the previously stated fact is that if permanent faulty links are allowed to recover then this could provoke instability, since a series of permanent failures and recoveries make the links behave as if experiencing transient faults. 
However, under some circumstances it is possible to allow permanent faulty links to recover and still guarantee stability.


\begin{theorem}
Any universally stable scheduling policy in $\cal DF^{\it e}_{\delta}$ remains universally stable when  permanent faulty links are allowed to recover, provided each link recovers only after all packets that have been re-routed because of its failure have reached their destinations.
\end{theorem}

\begin{proof}
From Lemma~\ref{lem:bou}, the number of re-routed packets is bounded up to the round when no permanent faulty link recovers. 
That means that the number of re-routed packets caused by the permanent failure of each link is also bounded. 
Since each link is allowed to recover only when all re-routed packets caused by the permanent failure of that link have reached their destination,  the number of re-routed packets in $\cal DF^{\it e}_{\delta}$ remains bounded, by Lemma~\ref{lem:bou}.
Therefore, we can apply Theorem~\ref{the:per}, which guarantees that the scheduling policy will remain universally stable.
\end{proof}

\section{Conclusion}

\label{sec:conclusion}

We study routing in the suitable adversarial frameworks.
We investigate how unexpected packet delays may affect routing's performance.
Packet delays represent either malfunctioning of the network's infrastructure, implemented below the network layer, or transient unavailability of nodes due to energy saving policies.

We assume that routing protocols are embedded into flow control and connection control mechanisms. 
These mechanisms react to packet delays by spreading the suitable information through the network with the goal to decrease packet injection rates.

We propose how to study stability of various classes of scheduling policies in such network settings.
To this end, we  propose a new adversarial model that has delays built into its machinery. 
The model we consider is an extension of the regular leaky-bucket model, which is determined only by the injection rate and burstiness. 

Each transmission that fails results in feedback, which abstracts the flow control and connection control mechanisms.
We treat this feedback as if given to the adversary, because it decreases the adversary's capability to inject packets.

We demonstrated that all scheduling policies stable in the 2-priority wireline adversarial model are also stable in the new proposed model. 
That includes such popular scheduling policies as FTG, NFS and SIS. 

The model is flexible enough to cover permanent failures. We found that stable scheduling policies in the setting with temporal failures remain stable when some failures are permanent.


\bibliography{adversary-feedback}

\bibliographystyle{abbrv}

\end{document}